%% file: main_SOSA.tex
\documentclass[11pt]{article}
\usepackage[utf8]{inputenc}
\usepackage{fullpage}

\usepackage[linktocpage]{hyperref}
\usepackage[dvipsnames]{xcolor}

\usepackage{amsmath}
\usepackage{amssymb}
\usepackage{amsthm}
\usepackage{amsfonts}
\usepackage{graphicx} 
\usepackage{wrapfig}
\usepackage{xspace}
\usepackage{thm-restate}
\usepackage{cleveref}

\usepackage[colorinlistoftodos]{todonotes}

\usepackage{xspace}
\newcommand{\LP}{(LP)\xspace}

\renewcommand{\epsilon}{\varepsilon}
\renewcommand{\phi}{\varphi}

\newcommand{\Odd}{Odd}

\DeclareMathOperator\dist{dist}
\DeclareMathOperator\LPopt{LP}

\newcounter{mnote}[section]

\definecolor{darkgreen}{rgb}{0.2,0.8,0.55}

\definecolor{ForestGreen}{rgb}{0.0333,0.4451,0.0333}
\definecolor{DarkRed}{rgb}{0.65,0,0}
\definecolor{Red}{rgb}{1,0,0}
\hypersetup{
    unicode=false,          
    colorlinks=true,        
    linkcolor=BrickRed,          
    citecolor=OliveGreen,        
    filecolor=magenta,      
    urlcolor=cyan           
}

\usepackage{float}
\usepackage[dvipsnames]{xcolor}

\usepackage[edges]{forest}	
 \usepackage{tabularx}
\usepackage{tikz}
\usetikzlibrary{arrows.meta}
\usepackage{tikz-3dplot}
\usetikzlibrary{calc,arrows,automata,patterns,graphs,shapes,petri,decorations.pathmorphing,decorations.markings,decorations.pathreplacing} 
\tikzset{->-/.style={decoration={markings,mark=at position .5 with {\arrow{>}}},postaction={decorate}}}
\tikzset{vertex/.style={draw,circle,inner sep=0pt,minimum size=18pt},>=latex'}
\tikzset{c/.style={draw,circle,inner sep=0pt,minimum size=15pt},>=latex'}
\tikzset{dot/.style={draw,circle,inner sep=0pt,minimum size=3pt, fill=black},>=latex'}

\usepackage{subcaption}
\usepackage{wrapfig}
\usepackage{algorithm}
\usepackage[noend]{algpseudocode}
\usepackage{fullpage}

\newtheorem{corollary}{Corollary}
\newtheorem{theorem}{Theorem}[section]
\newtheorem{lemma}[theorem]{Lemma}

\usepackage{mathtools,amssymb}

\let\emptyset\varnothing
\let\cref\Cref

\title{Approximating Graphic Multi-Path TSP and Graphic Ordered TSP}

\author{
    \begin{tabular}[t]{c@{\extracolsep{3em}}c@{\extracolsep{3em}}c}
        Morteza Alimi & Niklas Dahlmeier & Tobias Mömke \\
        \small University of Augsburg & \small RWTH Aachen & \small University of Augsburg \\
        \\
        \multicolumn{3}{c}{
            \begin{tabular}{@{}c@{\extracolsep{3em}}c@{}}
                Philipp Pabst & Laura Vargas Koch \\
                \small RWTH Aachen & \small RWTH Aachen
            \end{tabular}
        }
    \end{tabular}
}

\date{ }

\begin{document}
\maketitle
\begin{abstract}
The path version of the Traveling Salesman Problem is one of the most well-studied variants of the ubiquitous TSP. Its generalization, the Multi-Path TSP, has recently been used in the  best known algorithm for path TSP by Traub and Vygen [Cambridge University Press, 2024]. The best known approximation factor for this problem is $2.214$ by B\"{o}hm, Friggstad, M\"{o}mke and Spoerhase {[}SODA 2025{]}.
In this paper we show that for the case of graphic metrics, a significantly better approximation guarantee of $2$ can be attained. Our algorithm is based on sampling paths from a decomposition of the flow corresponding to the optimal solution to the LP for the problem, and connecting the left-out vertices with doubled edges. The cost of the latter is twice the optimum in the worst case; we show how the cost of the sampled paths can be absorbed into it without increasing the approximation factor.    
Furthermore, we prove that any below-$2$ approximation algorithm for the special case of the problem where each source is the same as the corresponding sink 
yields a below-$2$ approximation algorithm for Graphic Multi-Path TSP.

We also show that our ideas can be utilized to give a factor $1.791$-approximation algorithm for Ordered TSP in graphic metrics, 
for which the aforementioned paper {[}SODA 2025{]} and Armbruster, Mnich and Nägele {[}APPROX 2024{]} give a $1.868$-approximation algorithm in general metrics.

\end{abstract}
\input{introduction}

\input{Algorithm}
\input{Correctness}
\input{Approximation_analysis}
\section*{Acknowledgment}
We would like to thank Nathan Klein for suggesting simplifications in the proof of Theorem~\ref{thm:main}.
We also thank the organizers of Santiago Summer Workshop on Combinatorial Optimization where this work started.
\bibliography{references}

\end{document}

%% file: introduction.tex

\section{Introduction}
 The Traveling Salesman Problem (TSP) and its natural variants are among the most fundamental and significant problems in combinatorial optimization. 
In the classical version of TSP, we are given an undirected complete graph $G=(V,E)$ with nonnegative edge weights $c\colon E \rightarrow \mathbb{R}_{\ge 0}$, satisfying the triangle inequality.\footnote{The original formulation of TSP allows arbitrary cost functions. From an approximation algorithms perspective, however, the triangle inequality is crucial and it will be satisfied for all problems studied in this paper.}
The goal is then to find a shortest tour visiting all vertices.
Here, we focus on the Graphic TSP and the Graphic Multi-Path TSP. 
Graphic TSP is an important special case of TSP, where 
$G=(V,E)$ is a connected graph with $c(e) = 1$ for each edge $e \in E$ and we consider the metric completion of $G$.
Multi-Path TSP is a generalization of the classical TSP problem.
In the Multi-Path TSP, we are given a TSP instance and additionally a set of $k$ commodities, each represented by a vertex tuple $\{(s_1, t_1), (s_2, t_2), \ldots, (s_k, t_k)\}$. 
The task is to find a path $P_i$ from $s_i$ to $t_i$ for each $i \in [k]$ such that every vertex $v \in V$ is visited at least once by some path $P_i$. 
In this paper, we study the Graphic Multi-Path TSP, which means that we solve the Multi-Path TSP on a graphic instance. 

Our main result is the following Theorem~\ref{thm:main}.
\begin{restatable}{theorem}{maintheorem}\label{thm:main}
    There is a $2$-approximation algorithm for Graphic Multi-Path TSP.
\end{restatable}
We present a simple algorithm with a short but interesting analysis.
In the algorithm, we first solve a linear programming relaxation of a Multi-Path TSP IP. The resulting fractional solution is a flow which can be decomposed into paths and cycles. We then sample one path per commodity with probabilities determined by the decomposition. In the last step, we connect each of the so far unconnected vertices for a cost of two (exploiting the fact that we are in the graphic setting).
As a corollary, our algorithm shows that the integrality gap of the mentioned LP is bounded from above by $2$.

For the graphic case, our algorithm improves on the best known approximation algorithm for the general Multi-Path TSP \cite{bohm2025approximating} which has an approximation factor of $1+2 e^{-1/2} < 2.214$. 
Both algorithms use similar approaches, building on the same LP formulation. The crucial difference is that instead of sampling a path, the general version from \cite{bohm2025approximating} uses the Bang-Jensen Lemma \cite{bang1995preserving,post2014linear} to sample an arborescence for each commodity (see also \cite{AMR24, AMN24, BKN24,BN23,FS17} for applications of the lemma in the context of TSP). This technical framework is used to ensure that every vertex is sampled with the same probability.  
In the general case, this is needed to bound the reconnection cost, however, it has the disadvantage that a costly parity correction is required. 

Conceptually, Multi-Path TSP consists of two parts. First, we have to connect $s_i$ and $t_i$ for each commodity $i \in [k]$ and second we have to cover all the vertices. 
These two problems are closely connected. The more vertices are already covered by the paths, the fewer vertices have to be connected in the second step. At the same time, however, we want to avoid covering the same vertex too often by different paths.
It will turn out that after choosing a reasonable set of paths, it is easy to cover the remaining vertices.

Furthermore, we provide a $(1+\frac{e}{2e-2}) < 1.791$-approximation algorithm for Graphic Ordered TSP, which can be interpreted as a special case of Graphic Multi-Path TSP. 
While the overall approach is the same approach as in \Cref{thm:main}, we use a different parity correction argument. 
In the Graphic Multi-Path TSP algorithm, doubling the reconnection edges can be considered as a parity correction step. In the Graphic Ordered TSP algorithm, we can choose between this option and the classical parity correction argument, via a join of the odd degree vertices.
This improves on the best known approximation factor for general Ordered TSP which is $\frac{3}{2}+\frac{1}{e} <  1.868$, see \cite{AMN24, bohm2025approximating}. 
\begin{restatable}{theorem}{orderedtsptheorem}\label{thm:ordered}
    There is a $(1+\frac{e}{2e-2})$-approximation algorithm for Graphic Ordered TSP. 
\end{restatable}

Lastly, in \Cref{sec:reduction}, we relate Graphic Multi-Path TSP to the Graphic Uncapacitated Multi-Depot VRP. This is a special case of Graphic Multi-Path TSP, where for each commodity $i \in [k]$ we have $s_i=t_i$.
By doubling a terminal-rooted spanning forest, we easily obtain a $2$-approximation for this problem.
We show that any improvement in the approximation ratio of $2$ for Graphic Uncapacitated Multi-Depot VRP leads to an improved approximation ratio for Graphic Multi-Path TSP. Intuitively, this holds since \Cref{alg:main} is bad in the case where the average distance between $s_i$ and $t_i$ is small and the reconnection step incurs high cost. 
\begin{restatable}{theorem}{vrp}\label{thm:vrp}
Suppose there is a $(2-\delta)$-approximation algorithm $\text{ALG}_{\text{VRP}}$ for Graphic Uncapacitated Multi-Depot VRP, for some constant $\delta>0$. 
Then there is also an approximation algorithm with factor
$2-f(\delta)$ for Graphic Multi-Path TSP, where $f$ is a strictly positive function. 
\end{restatable}



\subsection{Further Related Work}

The TSP and its variants are well-studied from both a theoretical and a practical perspective. In this work, we will focus on approximation results and refer to the book \cite{Traub_Vygen_2024} for an exhaustive and up-to-date overview of existing results on many different variants of TSP.

For over 40 years, the best known approximation ratio for TSP was $3/2$ as achieved by the algorithm of Christofides~\cite{Christofides1976WorstCaseAO} and by Serdjukov~\cite{Ser78}. This was improved to $3/2-\epsilon$ \cite{10.1145/3406325.3451009} for some $\epsilon > 10^{-36}$. 
However, for the graphic case better approximation algorithms are known. 
Oveis Gharan, Saberi, and Singh gave a slight improvement over 3/2~\cite{gharan2011randomized},
shortly followed by the 1.461-approximation algorithm of
Mömke and Svensson~\cite{momke2011approximating}.
This was further improved first by Mucha~\cite{mucha2014approximation} and then by Seb\H{o} and Vygen~\cite{sebHo2014shorter} to 1.4.
By applying a reduction from Multi-Path TSP to TSP for a constant number of terminal pairs \cite{Traub_Vygen_2024,traub2021reducing}, an $\alpha$-approximation algorithm for Graphic TSP implies an $(\alpha+\epsilon)$-approximation algorithm Graphic Multi-Path TSP if the number of terminal pairs is bounded by a constant.
For the prize-collecting versions of Multi-Path TSP and Ordered TSP, where terminal pairs may be omitted for a penalty cost, there are a 2.41-approximation algorithm and a 2.097-approximation algorithm, respectively~\cite{AMR24}.

%% file: Algorithm.tex
\section{An LP-Formulation}

The Graphic Multi-Path TSP combines Multi-Path TSP and Graphic TSP. Formally, the problem is defined as follows:\\[1em]
\textsc{Graphic Multi-Path TSP}\vspace*{-0.2cm}
\begin{description}
  \item[\textbf{Input:}] We are given an undirected graph $G = (V,E)$ and a list of pairwise distinct source-sink-pairs $\{(s_1, t_1), (s_2, t_2), \ldots, (s_k, t_k)\}$ with $s_i, t_i \in V$ for all $i \in [k]$.
    \vspace*{-0.2cm}
  \item[\textbf{Question:}] For each pair $(s_i, t_i)$, find an $s_i$-$t_i$ walk $P_i$ such that $\bigcup_{i\in[k]}P_i = V$, and the total number of edges used by walks $\sum_{i\in[k]} |P_i|$ is minimized.
\end{description}
\vspace{1em}
 Note that we will not work on the metric completion but on the original graph throughout the whole paper. Thus, we can denote the cost of path $P$ by $c(P)\coloneqq \sum_{e \in P} 1 =|P|$, so every edge contributes exactly 1.

Throughout the paper, we will assume that a given instance is feasible. This means every $s_i$ is connected to $t_i$ and every vertex $v \in V$ is connected to at least one $s_i$ for all $i \in [k]$. This assumption is without loss of generality, as we can detect infeasibility by checking for the necessary connectedness in polynomial time. Moreover, we assume without loss of generality that $G$ is connected, as each connected component forms an independent subinstance.
We reformulate the problem in an equivalent way by bidirecting all edges. This yields a newly constructed graph $G'=(V,A)$.
As we assume that $G$ is connected, $G'$ is strongly connected. 
We will work with $G'$ rather than $G$ for the remainder of this paper. Given a solution for this bidirected version of the problem on $G'$, we can turn it back into a solution for the original problem on $G$ by removing all edge directions.

We now consider the following LP formulation 
for the (Graphic) Multi-Path TSP on $G'$. 
For a set of vertices $U$, we define $\delta^-(U) \coloneqq \{a = (u,v) \in A \colon u \notin U \text{ and } v \in U\}$, $\delta^+(U) \coloneqq \delta^-(V\setminus U)$, and for a vertex $v \in V$, we use the shorthand notation $\delta^-(v) \coloneqq \delta^-(\{v\})$ and $\delta^+(v) \coloneqq \delta^+(\{v\})$. Furthermore, we define
$T \coloneqq \{t_i \colon i \in [k]\}$ and $S \coloneqq \{s_i \colon i \in [k]\}$. 

Let $x_{i,a}$ be the flow on arc $a$ via commodity $i$, and $x_i(F) \coloneqq \sum_{a \in F} x_{i,a}$ for any set of arcs $F \subseteq A$.
Moreover, let $z_{i,v}$ be the total outflow of vertex $v$ of commodity $i$ and let $z_v \coloneqq \sum_{i \in [k]} z_{i,v}$ be the total outflow of vertex $v$.

\input{LP}

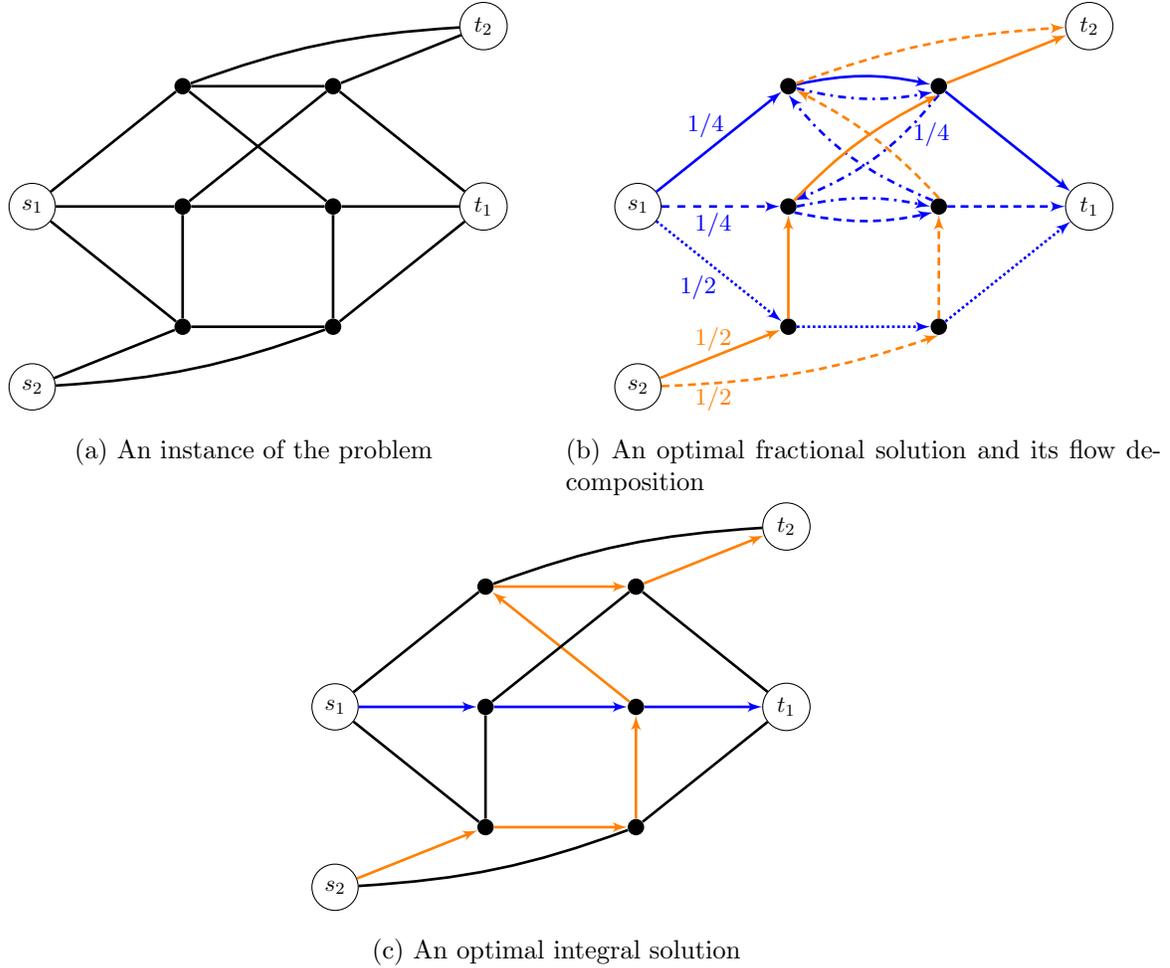
\begin{figure}[!tb]
    \centering
    \subcaptionbox{An instance of the problem \label{fig:instance}}[0.48\textwidth]{
    \begin{tikzpicture}[yscale = 0.8]

    \node[shape = circle, scale = 0.8, draw = black] (S1) at (0,-1) {$s_1$};
    \node[shape = circle, scale = 0.8, draw = black] (S2) at (0,-4) {$s_2$};
    \node[shape = circle, scale = 0.8, draw = black] (T1) at (6,-1) {$t_1$};
    \node[shape = circle, scale = 0.8, draw = black] (T2) at (6,2) {$t_2$};

    \node[shape = circle, fill = black, scale = 0.6] (1) at (2,1) {};
    \node[shape = circle, fill = black, scale = 0.6] (2) at (2,-1) {};
    \node[shape = circle, fill = black, scale = 0.6] (3) at (4,1) {};
    \node[shape = circle, fill = black, scale = 0.6] (4) at (4,-1) {};
    \node[shape = circle, fill = black, scale = 0.6] (5) at (2, -3) {};
    \node[shape = circle, fill = black, scale = 0.6] (6) at (4, -3) {};
    \node[] (L2) at (.75, -4.2) {\phantom{\footnotesize{1/2}}};

    \begin{scope}[line width = 1pt]
    \draw (S1) -- (1);
    \draw (S1) -- (2);
    \draw (S2) -- (5);
    \path (S2) edge[bend right = 10] (6);
    \draw (1) -- (3);
    \draw (1) -- (4);
    \draw (2) -- (3);
    \draw (2) -- (4);
    \draw (3) -- (T1);
    \draw (4) -- (T1);
    \draw (2) -- (5);
    \draw (4) -- (6);
    \draw (S1) -- (5);
    \draw (5) -- (6);
    \draw (6) -- (T1);
    \path (1) edge[bend left = 10] (T2);
    \draw (3) -- (T2);
    \end{scope}
    \end{tikzpicture}
    }
    \subcaptionbox{An optimal fractional solution and its flow decomposition \label{fig:lpsol}}[.48\textwidth]{
    \begin{tikzpicture}[yscale = 0.8]

    \node[shape = circle, scale = 0.8, draw = black] (S1) at (0,-1) {$s_1$};
    \node[shape = circle, scale = 0.8, draw = black] (S2) at (0,-4) {$s_2$};
    \node[shape = circle, scale = 0.8, draw = black] (T1) at (6,-1) {$t_1$};
    \node[shape = circle, scale = 0.8, draw = black] (T2) at (6,2) {$t_2$};

    \node[shape = circle, fill = black, scale = 0.6] (1) at (2,1) {};
    \node[shape = circle, fill = black, scale = 0.6] (2) at (2,-1) {};
    \node[shape = circle, fill = black, scale = 0.6] (3) at (4,1) {};
    \node[shape = circle, fill = black, scale = 0.6] (4) at (4,-1) {};
    \node[shape = circle, fill = black, scale = 0.6] (5) at (2, -3) {};
    \node[shape = circle, fill = black, scale = 0.6] (6) at (4, -3) {};

    \begin{scope}[line width = 1pt, color = blue]
    \draw[->] (S1) -- (1);
    \draw[dashed, ->] (S1) -- (2);
    \path (1) edge [->, bend left = 15] (3);
    \path (2.south east) edge [->, densely dashed, bend right = 15] (4.south west);
    \draw[->] (3) -- (T1);
    \draw[densely dashed, ->] (4) -- (T1);
    \draw[->, densely dotted] (S1) -- (5);
    \draw[->, densely dotted] (5) -- (6);
    \draw[->, densely dotted] (6) -- (T1);
    \path (1) edge[->, bend right = 15, dash dot] (3.south west);
    \path (2.east) edge[->, bend left = 15, dash dot] (4.west);
    \path (3.south) edge[->, bend left = 15, dash dot] (2.north east);
    \path (4.north west) edge[->, bend left = 15, dash dot] (1.south);
    \end{scope}

    \node[] (L1) at (0.9, 0.35) {\textcolor{blue}{\footnotesize{$1/4$}}};
    \node[] (L2) at (1, -1.3) {\textcolor{blue}{\footnotesize{$1/4$}}};
     \node[] (L1) at (0.8, -2.35) {\textcolor{blue}{\footnotesize{$1/2$}}};
    \node[] (L1) at (3.9, 0.2) {\textcolor{blue}{\footnotesize{$1/4$}}};

    \begin{scope}[line width = 1pt, color = orange]
    \path (S2) edge[->] (5);
    \path (S2) edge[->, bend right = 10, densely dashed] (6.south);
    \path (4.north) edge[->, densely dashed, bend right = 10] (1);
    \path (2) edge[->, bend left = 10] (3.south);
    \path (1) edge[->, bend left = 10, densely dashed] (T2);
    \path (3) edge[->] (T2);
    \draw[->] (5) -- (2);
    \draw[->, densely dashed] (6) -- (4);

    \node[] (L2) at (1, -3.2) {\textcolor{orange}{\footnotesize{$1/2$}}};
    \node[] (L2) at (1, -4.2) {\textcolor{orange}{\footnotesize{$1/2$}}};
    \end{scope}
    \end{tikzpicture}
    }

    \subcaptionbox{An optimal integral solution \label{fig:ipsol}}[0.48\textwidth]{
    \begin{tikzpicture}[yscale = 0.8]

    \node[shape = circle, scale = 0.8, draw = black] (S1) at (0,-1) {$s_1$};
    \node[shape = circle, scale = 0.8, draw = black] (S2) at (0,-4) {$s_2$};
    \node[shape = circle, scale = 0.8, draw = black] (T1) at (6,-1) {$t_1$};
    \node[shape = circle, scale = 0.8, draw = black] (T2) at (6,2) {$t_2$};

    \node[shape = circle, fill = black, scale = 0.6] (1) at (2,1) {};
    \node[shape = circle, fill = black, scale = 0.6] (2) at (2,-1) {};
    \node[shape = circle, fill = black, scale = 0.6] (3) at (4,1) {};
    \node[shape = circle, fill = black, scale = 0.6] (4) at (4,-1) {};
    \node[shape = circle, fill = black, scale = 0.6] (5) at (2, -3) {};
    \node[shape = circle, fill = black, scale = 0.6] (6) at (4, -3) {};
    \node[] (L2) at (.75, -4.2) {\phantom{\footnotesize{1/2}}};

    \begin{scope}[line width = 1pt]
    \path (S2) edge[->, color = orange] (5);
    \path (5) edge[->, color = orange] (6);
    \path (6) edge[->, color = orange] (4);
    \path (4) edge[->, color = orange] (1);
    \path (1) edge[->, color = orange] (3);
    \path (3) edge[->, color = orange] (T2);

    \path (S1) edge[->, color = blue] (2);
    \path (2) edge[->, color = blue] (4);
    \path (4) edge[->, color = blue] (T1);

    \draw[-] (S1) -- (5) -- (2) -- (3) -- (T1) -- (6);
    \draw[-] (S1) -- (1);
    \path (S2) edge[bend right = 10] (6);
    \path (1) edge[bend left = 10] (T2);
    
    \end{scope}
    \end{tikzpicture}
    }
    \caption{
    The fractional solution in (b) of the instance in (a) has an objective function value of 8. The first commodity sends flow along three paths and one cycle (marked by different types of blue arcs), and the second commodity sends flow along two paths drawn in orange. We have $\sum_{i\in[k]} z_{i,v} = 1$ for each $v \in V \setminus T$, which implies that the fractional solution is optimal. The optimal integral solution is shown in (c), and has an objective function value of 9.}
    \label{fig:enter-label}
\end{figure}

An example instance of the problem, as well as its optimal fractional and integral solution, are depicted in \cref{fig:enter-label}.

We will denote the above LP by \LP and its optimal objective function value by $\LPopt$ (this always exists as we assume feasibility and the objective function is bounded by 0). We note that \LP is essentially equivalent to the LP used by Böhm et al.~\cite{bohm2025approximating}. In particular, it can be solved efficiently, as separation of constraints \eqref{U_and_v} is possible via minimum cut computations. 
To see that \LP is a relaxation for (Graphic) Multi-Path TSP, observe that an optimal integral solution is a feasible solution to \LP by directing all $s_i$-$t_i$ paths and sending a flow of $1$ along each path and defining $z_{i,v}$ as in \eqref{z_and_x}.

\section{Graphic Multi-Path TSP}
In this section, we present our 2-approximation algorithm and thus prove \Cref{thm:main}. 
We start with a description of the algorithm and then provide pseudo-code in
 Algorithm~\ref{alg:main}.
The algorithm first computes an optimal fractional solution to \LP. 
The solution restricted to each commodity is a flow. By a greedy approach, it can thus be decomposed into at most $|A|$
paths and cycles, i.e., a polynomial number \cite[Theorem 11.1]{schrijver2003combinatorial}.
For every commodity $i \in [k]$, we denote by $P_1^i,\ldots, P_{\ell_i}^i$ the paths and by $C_1^i,\ldots, C_{\ell^{'}_i}^i$ the cycles in such a decomposition, 
with flow values $\lambda_1^i, \ldots, \lambda_{\ell_i}^i$ and $\mu_1^i, \ldots, \mu_{\ell^{'}_i}^i$, respectively.
Note that $\sum_j \lambda_j^i = 1$ for all commodities $i \in [k]$ with $s_i \neq t_i$ by constraints \eqref{outflow_s} and \eqref{inflow_t}. 
Moreover, let $V(P)$ denote the vertices of path $P$. 
For each commodity $i$ and for each vertex $v\neq t_i$ we have 
 $\sum_{j\colon v \in V(P^i_j)} \lambda^i_j \le z_{i,v}$.

In our algorithm, we sample exactly one path from the path decomposition for each commodity $i \in [k]$ with $s_i \neq t_i$, where path $P^i_j$ has probability $\lambda_j^i$ of being sampled. Formally, for each such commodity, we partition the interval $[0,1]$ into $\ell_i$ intervals with length $\lambda_1^i, \ldots, \lambda_{\ell_i}^i$ respectively and sample $y \in [0,1]$ uniformly at random.\footnote{Note, that the ordering of the partition does not matter as we sample from a uniform distribution.} We then choose the path corresponding to the unique interval to which $y$ belongs. 

Lastly, we connect all vertices that have not been covered by any path. We can do so at a cost of two per uncovered vertex as follows.
Denote by $X$ the vertices in the union of all sampled paths. 
Let $Y \coloneqq V\setminus X$ be the (so far) disconnected vertices. While $Y$ is nonempty, choose an arbitrary arc $a = (v,w)$ with $v \in Y$ and $w \in V \setminus Y$; since $G'$ is strongly connected, such an arc $a$ exists. 
We connect $v$ by adding the arcs $(w,v)$ and $(v,w)$ to a walk containing $w$. Then we remove $v$ from $Y$ and repeat this procedure until all vertices have been removed from $Y$.
This increases the cost of the solution by at most $2 \cdot |Y|$.

This algorithm yields a set of walks covering all vertices and connecting all pairs $s_i$-$t_i$ for $i \in [k]$. An example run of the algorithm is depicted in \cref{fig:algo}.

%
\begin{algorithm}[t]
    \caption{Algorithm for Graphic Multi-Path TSP}
    \label{alg:main}
    \begin{algorithmic}[0] 
            \State \textbf{Input:} Graph $G'=(V,A)$, tuples $\{(s_1, t_1), (s_2, t_2), \ldots, (s_k, t_k)\}$.
            \State \textbf{Output:} Collection $(F_i)_{i \in [k]}$ of directed walks covering all vertices, where $F_i$ is connecting $s_i$ to $t_i$ for all $i\in [k]$.
            \State Compute an optimal fractional solution $(x,z)$ to \LP. 
            \State Decompose each $x_i$ into paths $P_1^i,\ldots, P_{\ell_i}^i$ and cycles $C_1^i,\ldots, C_{\ell^{'}_i}^i$.\\
            \vspace{0.2em}
            \textbf{Sampling Phase:}
            \For{$i \in [k]$ with $s_i \neq t_i$} 
            \State{Partition the interval $[0,1]$ into $\ell_i$ intervals with length $\lambda_1^i, \ldots, \lambda_{\ell_i}^i$ respectively.}
            \State{Sample $y \in [0,1]$ uniformly at random.}
            \State{Choose path $P_j^i$ corresponding to the unique interval to which $y$ belongs.}
            \State{Let $F_i \coloneqq P_j^i$.}
            \EndFor 
            \For{$i \in [k]$ with $s_i = t_i$}
             \State{Let $F_i \coloneqq \{s_i\}$.}
            \EndFor \\
            \vspace{0.2em}
            \textbf{Reconnection Phase:}
            \State Set $X$ as the set of all vertices in the union of the sampled paths.
            \State Set $Y \coloneqq V\setminus X$.
            \While{$Y \neq \emptyset$}
                \State Choose an arc $a = (v,w) \in A$ with $v \in Y$, and $w \in V \setminus Y$.
                \State Add both $(v,w)$ and $(w,v)$ to $F_i$ for some $F_i$ containing $w$.
                \State Update $Y \coloneqq Y \setminus \{v\}.$
            \EndWhile \\
        \Return $(F_i)_{i\in[k]}$
    \end{algorithmic}
\end{algorithm}

\begin{figure}[H]
    \centering
        \centering
    \begin{tikzpicture}[yscale = 0.8]

    \node[shape = circle, scale = 0.8, draw = black] (S1) at (0,-1) {$s_1$};
    \node[shape = circle, scale = 0.8, draw = black] (S2) at (0,-4) {$s_2$};
    \node[shape = circle, scale = 0.8, draw = black] (T1) at (6,-1) {$t_1$};
    \node[shape = circle, scale = 0.8, draw = black] (T2) at (6,2) {$t_2$};

    \node[shape = circle, fill = black, scale = 0.6] (1) at (2,1) {};
    \node[shape = circle, fill = black, scale = 0.6] (2) at (2,-1) {};
    \node[shape = circle, fill = black, scale = 0.6] (3) at (4,1) {};
    \node[shape = circle, fill = black, scale = 0.6] (4) at (4,-1) {};
    \node[shape = circle, fill = black, scale = 0.6] (5) at (2,-3) {};
    \node[shape = circle, fill = black, scale = 0.6] (6) at (4,-3) {};

    \begin{scope}[line width = 1.4pt]
    \path (S1) edge[->, blue] (1);
    \path (S2) edge[bend right = 10, ->, orange] (6.south west);
    \path (6) edge[->, orange] (4);
    \path (1) edge[->, blue] (3);
    \path (3) edge[->, blue] (T1);
    \path (1) edge[bend left = 10, ->, orange] (T2);
    \path (4) edge[->, orange] (1);
    \path (2) edge[->, bend left = 15, densely dotted, orange] (4);
    \path (4) edge[->, bend left = 15, densely dotted, orange] (2);
    \path (5) edge[->, bend left = 15, densely dotted, orange] (2);
    \path (2) edge[->, bend left = 15, densely dotted, orange] (5);
    \end{scope}
    \end{tikzpicture}

    \caption{
    Run of the algorithm on the instance of Figure \ref{fig:instance}. The algorithm, when run on this instance, samples one path per commodity, in this case the solid blue path $F_1$ for commodity 1 and the solid orange path $F_2$ for commodity 2. After sampling, 
    two of the inner vertices are not covered by the union of all paths $F_i$.
    In the reconnection phase, these vertices get connected with a cost of 2 per vertex (marked by dotted arcs).
    \label{fig:algo}
    }
\end{figure}
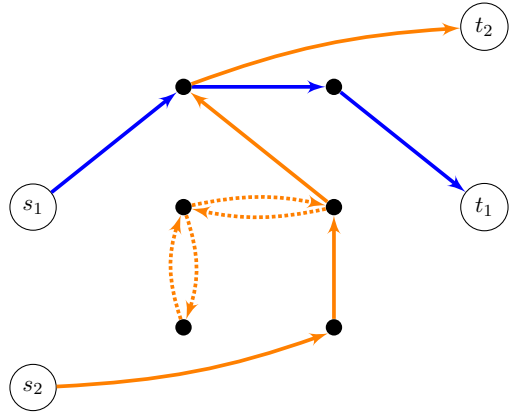

%% file: LP.tex
\begin{alignat}{3} \label{LP}
 & \text{minimize} & \sum_{i \in [k]} \sum_{a \in A} x_{i,a}&&& \\
 & \text{s.~t.} \quad& x_i(\delta^{+}(v)) &= x_i(\delta^{-}(v)), & \quad &\forall i \in [k],v \in V \setminus (\{s_i\} \Delta \{t_i\})\label{in=out}\\
                 && x_i(\delta^{+}(s_i)) - x_i(\delta^-(s_i)) & = 1, & \quad &\forall i \in [k], s_i\neq t_i\label{outflow_s} \tag{3a}\\
                 && x_i(\delta^{-}(t_i)) - x_i(\delta^+(t_i)) & = 1, & \quad &\forall i \in [k], s_i\neq t_i\label{inflow_t} \tag{3b} \\
                 && x_i(\delta^{+}(U)) & \ge x_i(\delta^{+}(v)), & \quad &\forall i \in [k], U \subseteq V \setminus \{t_i\}, v \in U \label{U_and_v} \tag{4}\\
                 && z_{i,v} &= x_i(\delta^{+}(v)), & \quad &\forall i \in [k], v \in V  \label{z_and_x} \tag{5}\\
                 && \sum_{i \in [k]} z_{i,v} & \geq 1, & \quad &\forall v \in V \setminus T \label{sum_of_z} \tag{6}\\
                 && x_{i,a} &\ge 0, & \quad &\forall i \in [k], a \in A \label{non_negativity}\tag{7}
\end{alignat}
\setcounter{equation}{7}

%% file: Correctness.tex
\subsection{Correctness}
To show correctness of \cref{alg:main}, we need to show that it terminates in polynomial time and that we can form $k$ walks from $s_i$ to $t_i$ for all $i \in [k]$ covering all vertices. 
\begin{restatable}{lemma}{correctness}\label{thm:correctness}
    \Cref{alg:main} terminates in polynomial time and outputs a valid solution for the bidirected version of Graphic Multi-Path TSP.
\end{restatable}
\begin{proof}
    The algorithm clearly terminates in polynomial time as the size of $Y$ decreases by one in each iteration of the \texttt{while}-loop, resulting in at most $O(n)$ iterations.

    Note that after the sampling phase, each $F_i$ connects $s_i$ to $t_i$ for $i \in [k]$ by a simple (possibly empty) path. 
    In the reconnection phase, we then make sure that all vertices are covered. 
    To see that this is possible, observe that since $G'$ is strongly connected, there is always a vertex $v \in Y$ such that there exists some $w \in V \setminus Y$  with $(v,w) \in A$. 
    As we maintain the invariant that all vertices $V \setminus Y$ are covered by at least one walk $F_i$, we have that there is some $i \in [k]$ with $w \in F_i$. 
    When adding a vertex $v$ to a walk $F_i$ containing $w$, we add the arcs $(w,v)$ and $(v,w)$ and thus maintain that $F_i$ is a walk.

\end{proof}
 

%% file: Approximation_analysis.tex
\subsection{Approximation Guarantee}

In this section, we prove that the approximation ratio of our algorithm is 2. 
For the analysis we introduce new variables as follows. Recall that $P_j^i$ is the $j$-th path of commodity $i$, of weight $\lambda_j^i$. 
We define
\begin{alignat*}{3}
z_{i,v}^P & \coloneqq \begin{cases} 
\sum_{j \colon  v \in V(P_j^i) } \lambda_j^i &\text{ if } v\neq t_i,\\
0 &\text{ if } v=t_i,
\end{cases}
& \qquad \qquad
z_v^P & \coloneqq  \sum_{i \in [k]} z_{i,v}^P. 
\end{alignat*}
Intuitively, $z_{i,v}^P$ denotes the outflow of vertex $v$ in commodity $i$ via paths (i.e., that outflow of $v$ which does not leave $v$ via cycles) and $z_v^P$ denotes the total outflow of vertex $v$ via paths.

\begin{lemma}
\label{lem:randomized_two_approx}
    \Cref{alg:main} computes a solution to the bidirected version of Graphic Multi-Path TSP of cost at most $2\cdot \LPopt$, and thus yields a randomized 2-approximation algorithm.
\end{lemma}

\begin{proof}
As we have proven correctness and polynomial running time in \Cref{thm:correctness} it only remains to prove the approximation guarantee. We sum the expected cost per vertex over all vertices $v \in V$ and split 
the expected cost of vertex $v$ into its \emph{sampling cost} and its \emph{reconnection cost}. The latter is the cost incurred if $v$ is not sampled and has to be reconnected.
The expected sampling cost of $v$ is the expected number of outgoing arcs $a \in \delta^+(v)$ that were sampled in the algorithm. This can be formulated as follows
\begin{align}
    \mathbb{E}[\text{sampling cost of $v$}] &= \sum_{i \in [k], v \neq t_i} \ \sum_{ j \colon v \in V(P^i_j) } \lambda_j^i =z_v^P.
    \label{samplingCost}
\end{align}
Thus, by summing over all vertices $v \in V$, we obtain the total cost that is incurred during the sampling phase.

To upper bound the expected reconnection cost of a vertex $v$, we need to bound the probability that $v$ is not part of any sampled path. First note, that for all terminal vertices $v\in S \cup T$, this probability is 0, as our sampling process ensures that all $(s_i, t_i)$-pairs are connected. For some non-terminal vertex $v \in V \setminus (S \cup T)$, the probability that $v$ is not covered by commodity $i \in [k]$ is $(1-z_{i,v}^P)$. Thus, the probability that $v \in V \setminus (S \cup T)$ is not covered by any commodity is $\prod_{i \in [k]} (1-z_{i,v}^P)$. Using the fact that $1-x \le e^{-x}$ for all $x\in\mathbb{R}$, we obtain the following bound
\[
    \mathbb{P}[\text{$v$ not covered by any path}] = \begin{cases} \prod_{i \in [k]} (1-z_{i,v}^P) \le \prod_{i \in [k]}e^{-z_{i,v}^P} = e^{-z_v^P}, & v \not\in S \cup T, \\ 0, & v\in S \cup T. \end{cases}
\]
As the reconnection cost for every unconnected vertex is 2 (which follows directly from \Cref{alg:main} and \Cref{thm:correctness}), this immediately implies
\begin{equation}
    \mathbb{E}[\text{reconnection cost of $v$}] \le \begin{cases} 2\cdot e^{-z_v^P}, & v\not\in S \cup T, \\ 0, & v\in S \cup T. \end{cases} \label{reconnectionCost}
\end{equation}
Putting \eqref{samplingCost} and \eqref{reconnectionCost} together we can bound the cost induced by vertex $v \in V \setminus (S \cup T)$ as follows

\begin{equation*}
z_v^P + 2e^{-z_v^P} \leq \begin{cases}
    2, \quad & \text{ if } z_v^P \leq 1,\\
    z_v^P+1 \leq 2 z_v^P, & \text{ if } z_v^P > 1.
\end{cases}
\end{equation*}
Recall, that the cost for vertex $v \in S \cup T$ is bounded by $z_v^P$. Moreover, let us define $z_v \coloneqq \sum_{i \in [k]}z_{i,v}$.
Now, summing over all vertices $v \in V$, we obtain the wished bound. 

\begin{align}
\label{eq:cost_of_tour}
    \mathbb{E}[\text{cost of tour}] &= \sum_{v \not\in S \cup T} (z_v^P + 2\cdot e^{-z_v^P}) + \sum_{v \in S \cup T} z_v^P  \nonumber \\
    &\leq \sum_{v \not\in S \cup T} \max \{2, 2z_v^P\} + \sum_{v \in S\cup T} z_v^P \nonumber \\
    &\leq 2 \sum_{v \not\in S \cup T} z_{v} + \sum_{v \in S\cup T} z_{v} \le 2 \sum_{v \in V} z_{v} = 2 \cdot \text{LP}.
\end{align} 
Recall that \text{LP} denotes the optimal objective function value of \LP. In the penultimate inequality we used the fact that $z_v^P \leq z_v$ for all $v \in V$, and $z_v \geq 1$ for all non-terminal vertices $v \not\in S \cup T$.
\end{proof}
Our analysis immediately implies an upper bound for the integrality gap of \LP.
\begin{corollary}\label{cor:ig}
    For Graphic Multi-Path TSP instances, 
    the integrality gap of \LP is bounded from above by $2$.
\end{corollary}
    

\subsection{Derandomization} 

The goal of this section is to derandomize \Cref{alg:main}. We do this by using the method of conditional expectations, to get a \emph{deterministic} $2$-approximation algorithm.
We will show the following lemma, which directly implies \Cref{thm:main}.

\begin{lemma}
    \Cref{alg:main} can be derandomized, while maintaining that the computed solution has cost at most $2\cdot \LPopt$.
\end{lemma}

\begin{proof}

We will prove the lemma by induction. 

The induction claim is that \Cref{alg:main} can be adapted in such a way that a path for commodities $i \in [h-1]$ is chosen deterministically in polynomial time, while the paths of all remaining commodities $i \in [k] \setminus [h-1]$ are sampled. More concretely, let $M \coloneqq \cup_{i=1}^{h-1} V(Q^i) \cup S \cup T$. We can choose paths 
$Q^i\in \{P_1^i,\ldots, P_{\ell_i}^i\}$, for $i \in [h-1]$ such that 
 \begin{equation*}
    \sum_{i=1}^{h-1} |Q^i| + 2\sum_{v\not\in M}\prod_{i=h}^{k} (1-z^P_{i,v}) + \sum_{i=h}^{k}\sum_{j=1}^{\ell_i}\lambda^i_j |P^i_j| \le 2 \cdot \LPopt,
\end{equation*}   
%
where the first term represents the cost of the already chosen paths $Q^i$ for $i \in [h-1]$, 
the next term contains the expected reconnection cost for non-terminal vertices not yet covered by any path, 
and the last term represents the expected cost of paths sampled for commodities $h$ to $k$. 

We induct over $h$, starting with $h=1$, where $[0]= \emptyset$. In this case, the induction claim follows directly by \Cref{lem:randomized_two_approx} (and its proof).

Now, suppose the induction claim holds for some $h$. 
  We fix a path $Q\in \{P_1^h,\ldots, P_{\ell_h}^h\}$. 
   Conditioned on choosing $Q$ as the path for commodity $h$ and sticking to the paths $Q^i$ for $i \in [h-1]$
   the expected cost of the output of the algorithm is
 \begin{equation}
 \label{eq:derand}
    |Q| + \sum_{i=1}^{h-1} |Q^i| + 2\sum_{v\not\in M \cup V(Q)}\prod_{i=(h+1)}^{k} (1-z^P_{i,v}) + \sum_{i=(h+1)}^{k}\sum_{j=1}^{\ell_i}\lambda^i_j |P^i_j|.
\end{equation}   
 

    By induction hypothesis, when sampling a path for commodity $h$ the expected cost of the output is at most
    $2 \cdot \text{LP}$.
   Hence, there exists some $Q^h\in \{P_1^h,\ldots, P_{\ell_h}^h\}$ such that the induction statement is satisfied. 
    As there are only linearly many options for commodity $h$ and we pick the path minimizing \eqref{eq:derand}, this can be executed in polynomial time.

    
        
\end{proof}

\section{Graphic Ordered TSP}
\label{sec:ordered_TSP}
In this section, we show that a variant of \Cref{alg:main} gives a $(1+\frac{e}{2e-2})$-approximation algorithm for the Graphic Ordered TSP. This improves upon the best known approximation factor for general Ordered TSP which is $\frac{3}{2}+\frac{1}{e} < 1.868$, see \cite{AMN24, bohm2025approximating}. 

Intuitively, in the Ordered TSP additionally to a set of vertices with a metric, we are given precedence constraints that determine a linear ordering on a subset of the vertices. The goal is to compute a TSP tour, i.e., a tour through all the vertices, where they appear in an order obeying the precedence constraints.
We will consider the graphic version of the problem and we will formally define it as a special case of Multi-Path TSP, as this fits well for our work.\\[1em]
\textsc{Graphic Ordered TSP}\vspace*{-0.2cm}
\begin{description}
  \item[\textbf{Input:}] An undirected, connected graph $G = (V,E)$ with unit edge costs and a list of $k$ terminals $(o_1,\ldots,o_k)$ with $o_i \in V$ for $i \in [k]$. \vspace*{-0.2cm}
  
  \item[\textbf{Question:}] For each $i \in [k]$, find a walk $P_i$ from $o_i$ to $o_{i+1}$ (using the notation $o_{k+1} = o_1$), such that $\bigcup_{i \in [k]} P_i = V$, and the total cost of all walks $\sum_{i\in[k]} |P_i|$ is minimized.
\end{description}
\vspace{1em}

The problem is a special case of Multi-Path TSP by setting $s_i=o_i$ and $t_i=o_{i+1}$, for all $1\le i \le k$. We will use this notation while describing the algorithm. 
While we allow in our definition that a vertex is visited multiple times on a tour, this is not always the case in the literature. However, in \cite{bohm2025approximating} it is shown how to shortcut a tour without violating the precedence constraints. Thus, when working with the metric completion, our formulation is equivalent to the one in which each vertex is just allowed to be visited once.

Our proposed algorithm for Graphic Ordered TSP works similarly  to \Cref{alg:main}, 
except that we do not double the edges while connecting the left-over vertices, 
and instead find a minimum size $\Odd$-join to correct the degree parity of the vertices (see \Cref{alg:otsp}). Recall that by $G' = (V,A)$ we denote the bidirected version of the undirected graph $G$.

\begin{algorithm}[h]
    \caption{Algorithm for Graphic Ordered TSP}
    \label{alg:otsp}
    \begin{algorithmic}[0] 
            \State \textbf{Input:} Graph $G'=(V,A)$, terminals $(o_1, \ldots, o_k)$.
            \State \textbf{Output:} Collection $(F_i)_{i \in [k]}$ of directed walks covering all vertices, where $F_i$ is connecting $o_i$ to $o_{i+1}$ for all $i\in [k]$.\\
            \vspace{-0.8em}
            \State Let $s_i \coloneqq o_i$ and $t_i \coloneqq o_{i+1}$ for all $i \in [k]$.
            \State Compute an optimal fractional solution $(x,z)$ to \LP. 
            \State Decompose each $x_i$ into paths $P_1^i,\ldots, P_{\ell_i}^i$ and 
cycles $C_1^i,\ldots, C_{\ell^{'}_i}^i$.\\
            \vspace{0.2em}
            \textbf{Sampling Phase:}
            \For{$i \in [k]$} 
            \State{Partition the interval $[0,1]$ into $\ell_i$ intervals with length $\lambda_1^i, \ldots, \lambda_{\ell_i}^i$ respectively.}
            \State{Sample $y \in [0,1]$ uniformly at random.}
            \State{Choose path $P_j^i$ corresponding to the unique interval to which $y$ belongs.}
            \State{Let $F_i \coloneqq P_j^i$.}
            \EndFor \\
            \vspace{0.2em}
            \textbf{Reconnection Phase:}
            \State Set $X$ as the set of all vertices in the union of the sampled paths.
            \State Set $Y \coloneqq V\setminus X$.
            \While{$Y \neq \emptyset$}
                \State Choose an arc $a = (v,w) \in A$ with $v \in Y$, and $w \in V \setminus Y$.
                \State Add $(v,w)$ to $F_i$ for some $F_i$ containing $w$.
                \State Update $Y \coloneqq Y \setminus \{v\}.$
            \EndWhile \\
        \vspace{0.2em}
        \textbf{Parity Correction Phase:}
        \State{Let $\Odd$ be the set of odd degree vertices.}
        \State{Let $J$ be a smallest $\Odd$-join in $G$.}
       
        \Return $(\cup_{i\in[k]}F_i)\cup J$
    \end{algorithmic}
\end{algorithm}

\orderedtsptheorem*
\begin{proof}
    The correctness of \Cref{alg:otsp} is easy to see.
   Furthermore, the algorithm has polynomial running time as the $\Odd$-join polytope is integral and thus finding the smallest $\Odd$-join is possible in polynomial time, see Chapter~29 in~\cite{schrijver2003combinatorial}. 
   
   In the following, we will analyze the approximation ratio of the algorithm.
    As in the proof of \Cref{lem:randomized_two_approx}, we can bound the cost of the sampling and reconnection phase by 
    \[
   \mathbb{E}[\text{cost of sampling and reconnecting}] = \sum_{v \not\in S \cup T} (z_v^P + e^{-z_v^P}) + \sum_{v \in S \cup T} z_v^P.
    \]
 The difference to the proof of \Cref{lem:randomized_two_approx} is that in the reconnection phase, instead of doubling edges for parity correction, we use a smarter approach.

To upper bound the cost of the parity correction step, we balance between two candidate $\Odd$-joins, to obtain a bound
on the cost of a smallest $\Odd$-join $J$. 
    Let $(x,z)$ be an optimal solution to \LP, and define $y\coloneqq \frac{1}{2}x$. 
    By the analysis of Wolsey~\cite{Wol80}, $y$ is in the $\Odd$-join polytope (see also \cite{bohm2025approximating}).  As the $\Odd$-join polytope is integral, this yields an upper bound of $ \frac{1}{2} \LPopt$ for the cost of the parity correction.
    
Also the edges added in the reconnection phase, which we will denote as $L$, comprise a valid $\Odd$-join, i.e., their incidence vector $\chi_L$ is also in the $\Odd$-join polytope. 
    If, instead of $J$, we use $L$ for parity correction (i.e.,\ run \cref{alg:main}), the expected cost of the parity correction is
    $\sum_{v \not\in S\cup T}  e^{-z_v^P}$.

    For sake of the analysis we will now balance between these two costs by choosing a convex combination of $y$ and $\chi_L$. Let $\alpha \in [0,1]$ and define a new vector in the $\Odd$-join polytope, namely $\alpha  \chi_L + (1-\alpha) y$. 

    The cost of this vector for parity correction is 
\begin{equation*}
    \mathbb{E}[\text{cost of parity correction}] = \sum_{v\not\in S \cup T}\left(\alpha \cdot  e^{-z_v^P} \right)  + (1-\alpha) \cdot \frac{\LPopt}{2}.
\end{equation*}

Now, we bound the total expected cost of the algorithm. We have
    
\begin{align*}
   \mathbb{E}[\text{cost }&\text{of tour}] \leq \sum_{v\not\in S \cup T}\left(z_v^P + (1+\alpha) \cdot e^{-z_v^P} \right) + \left(\sum_{v \in S \cup T} z_v^P\right) + (1-\alpha) \cdot \frac{\LPopt}{2}\\
    &= \sum_{v\not\in S \cup T}\left(z_v^P + (1+\alpha) \cdot e^{-z_v^P} + (1-\alpha)\cdot \frac{z_v}{2}  \right) + \sum_{v \in S \cup T} \left(z_v^P + (1-\alpha)\cdot\frac{z_v}{2}\right)\\
    &\leq \sum_{v\not\in S \cup T} \left( \max\left\{(1+\alpha), \left(1+ \frac{1+\alpha}{e}\right)\right\}  \cdot z_v + (1-\alpha) \cdot \frac{z_v}{2} \right) + \sum_{v \in S \cup T} \left(z_v^P + (1-\alpha)\cdot\frac{z_v}{2}\right).
\end{align*}

In the last step, we exploited that $g(z_v^P) \coloneqq z_v^P + (1+\alpha) \cdot e^{-z_v^P}$ is a convex function and thus we have $g(z_v^P) \leq \max\{ g(0), g(1) z_v^P\}$ for $z_v^P \in \mathbb{R}_{\geq 0}$ and, additionally, we used that $z_v^P \leq z_v$ and $1 \leq z_v$.
This cost is minimized for $\alpha= \frac{1}{e-1}$. For this choice of $\alpha$, we consider the two sums in the last term separately. For the first sum, we obtain
\begin{align*}
    & \sum_{v\not\in S \cup T} \left( \max\left\{(1+\alpha), \left(1+ \frac{1+\alpha}{e}\right)\right\}  \cdot z_v + (1-\alpha) \cdot \frac{z_v}{2} \right) \\
    & \leq \sum_{v\not\in S \cup T} \left( \max\left\{1+\frac{1}{e-1}, 1+ \frac{1+\frac{1}{e-1}}{e}\right\} \cdot z_v + \left(1-\frac{1}{e-1}\right) \cdot \frac{z_v}{2} \right)\\
    & \leq \sum_{v\not\in S \cup T} \left( \max\left\{\frac{e}{e-1} , \frac{e}{e-1} \right\}  \cdot z_v + \frac{e-2}{2(e-1)} \cdot z_v \right)\\
    &= \sum_{v\not\in S \cup T}  \left(1+ \frac{e}{2e-2} \right) \cdot z_v.
\end{align*}
For the second sum, we calculate
\[\sum_{v \in S \cup T} \left(z_v^P + (1-\alpha)\cdot\frac{z_v}{2}\right) \leq \sum_{v \in S \cup T} \left(z_v + (1-\alpha)\cdot\frac{z_v}{2}\right) = \sum_{v \in S\cup T} \left(1 + \frac{e-2}{2e-2}\right)\cdot z_v.\]
Thus, adding up the two sums we obtain the following bound for the total cost of the tour
\begin{align*}
\mathbb{E}[\text{cost of tour}] &\leq \sum_{v\not\in S \cup T}  \left( 1+ \frac{e}{2e-2} \right) \cdot z_v + \sum_{v \in S\cup T} \left(1 + \frac{e-2}{2e-2}\right)\cdot z_v \\
& \leq \left(1 + \frac{e}{2e-2}\right) \sum_{v \in V} z_v = \left(1 + \frac{e}{2e-2}\right)\cdot \LPopt.
\end{align*}
\end{proof}

\section{Reducing Graphic Multi-Path TSP to Graphic Uncapacitated Multi-Depot VRP} \label{sec:reduction}
In this section, we utilize \cref{alg:main} to show that
a below-two approximation algorithm for the special case of Graphic Multi-Path TSP where 
$s_i=t_i$ for all $i \in [k]$ (which we term  \textbf{Graphic Uncapacitated Multi-Depot VRP})
gives rise to a below-two approximation algorithm for the Graphic Multi-Path TSP. 
To do so, given an instance of Graphic Multi-Path TSP we distinguish between two cases: If the total distance between all source-sink pairs is small, 
we create an instance of Graphic Uncapacitated Multi-Depot VRP, by ignoring the sink nodes of all commodities. We run the below-two approximation on the instance and then connect all $(s_i,t_i)$-pairs with shortest paths. On the other hand, if the total distance is large, we run \cref{alg:main} and show that in this case this algorithm already yields a below-two approximation.

\vrp*
%

\begin{proof}
Let $I$ be an instance of Graphic Multi-Path TSP consisting of $G=(V,E)$ and source-sink pairs $\{(s_1,t_1),\ldots, (s_k,t_k) \}$. We define the 
\textit{associated} instance $I' = (G, \{(s_1,s_1),\ldots, (s_k,s_k) \})$ to be the same as $I$, except that for 
each $i$, we set both terminals to be equal to $s_i$;
hence $I'$ is an instance of the Graphic Uncapacitated Multi-Depot VRP.
Also define $D(I) \coloneqq \sum_{i\in [k]} \dist(s_i,t_i)$ to be the sum of $s_i$-$t_i$ distances, which are defined as the minimal number of edges on any path connecting $s_i$ and $t_i$, i.e., $\dist(s_i,t_i) \coloneq \min \{c(P) \colon P \text{ is $s_i$-$t_i$-path}\} $. 

Observe that $|\text{OPT}(I)-\text{OPT}(I')| \le D(I)$, because if we take any solution for $I$ and we add shortest $s_i$-$t_i$ paths,
we get a solution for $I'$, and vice versa. 

Suppose $\text{ALG}_{\text{VRP}}$ is a $(2-\delta)$-approximation algorithm for Graphic Uncapacitated Multi-Depot VRP, for some constant $\delta>0$. 
Then, we can use \cref{alg:reduction} to obtain a $2-f(\delta)$-approximation for Graphic Multi-Path TSP.

Run \cref{alg:main} on $I$ to get candidate solution $F_1$. Also run $\text{ALG}_{\text{VRP}}$ on the associated instance $I'$, and add the shortest $s_i$-$t_i$ paths for all $i \in [k]$ to the output to get candidate solution $F_2$. 
We claim that the cost of the cheaper of $F_1$ and $F_2$ is at most 
$(2-f(\delta))\cdot \text{OPT}(I)$ for some strictly positive function $f$. 

\begin{algorithm}[t]
    \caption{Multi-Path TSP using Graphic Uncapacitated Multi-Depot VRP}
    \label{alg:reduction}
    \begin{algorithmic}[0] 
            \State \textbf{Input:} Graph $G=(V,E)$, tuples $\{(s_1, t_1), (s_2, t_2), \ldots, (s_k, t_k)\}$.
            \State \textbf{Output:} Collection $(F_i)_{i \in [k]}$ of walks covering all vertices, where $F_i$ is connecting $s_i$ to $t_i$ for all $i\in [k]$.
            \State Compute a solution $F_1$ for the bidirected version of $G$ and $\{(s_1, t_1), (s_2, t_2), \ldots, (s_k, t_k)\}$ using \Cref{alg:main}.
            \State Compute a solution $\widetilde F_2$ for $I' = (G=(V,E), \{(s_1,s_1),\ldots, (s_k,s_k) \})$ using the black box algorithm $\text{ALG}_{\text{VRP}}$.
            \State Add shortest $s_i$-$t_i$-path for all $i \in [k]$ to  $\widetilde F_2$ to obtain $F_2$.\\
        \Return $F$, which is the cheaper solution of $F_1$ and $F_2$.
    \end{algorithmic}
\end{algorithm}

We distinguish three cases depending on $D(I)$ and $\sum_{v\in S\cup T}z_v$. \\

\noindent \textbf{Case 1: $\sum_{v \in S\cup T} z_{v} \ge \frac{\delta}{16}\text{OPT}(I)$}.

In this case, we can bound the cost of solution $F_1$ obtained by \Cref{alg:main} by using the last line of \Cref{eq:cost_of_tour}
\begin{align*}
    \mathbb{E}[c(F_1)] \leq 2 \sum_{v \not\in S \cup T} z_{v} + \sum_{v \in S\cup T} z_{v} &= 2 \sum_{v \in V} z_{v} - \sum_{v \in S\cup T} z_{v} \le 2 \cdot \text{LP} - \frac{\delta}{16}\text{OPT}(I).
\end{align*}
\noindent \textbf{Case 2: $D(I)\le \frac{\delta}{4}\text{OPT}(I)$}. 

In this case, we bound the cost of the solution $F_2$ obtained from $\text{ALG}_{\text{VRP}}$ together with the cost of connecting the $(s_i,t_i)$-pairs. We have
\begin{align*}
    c(F_2) \le (2-\delta)\text{OPT}(I') + D(I) \le (2-\delta)(\text{OPT}(I)+D(I)) + D(I) \\ %
            < (2-\delta)\text{OPT}(I) + 3D(I) \le \left(2-\frac{\delta}{4}\right)\text{OPT}(I).
\end{align*}
\textbf{Case 3: $D(I) > \frac{\delta}{4}\text{OPT}(I)$}.

In this case, we analyze the cost of the solution $F_1$ obtained by \Cref{alg:main}. Note that instead of summing over all vertices, we can sum over all paths in the given path decomposition $P_j^{i}$ with corresponding weights $\lambda^i_j$.
Hence, for each $i \in [k]$, 
\[\sum_{v \in V} z_{i,v}^P \ge \sum_{j \in [\ell_i]} \lambda^i_j|P^i_j| \ge \sum_{j \in [\ell_i]} \lambda^i_j \dist(s_i,t_i) = \dist(s_i,t_i).\]
This means
\begin{equation*}\label{eq:delta4}
\sum_{v\in V} z_{v}^P \ge \sum_{i \in [k]}\dist(s_i,t_i) = D(I) \ge \frac{\delta}{4}\text{OPT}(I).
\end{equation*}
Define 
$V_P = \{  v\in V \setminus (S \cup T) \colon z^P_{v} \ge \frac{\delta}{16}\}.$ 
We distinguish two sub-cases depending on whether the set $V_P$ is large or small. \\

\noindent \textbf{Case 3.1: $|V_P| \ge \frac{\delta}{16}\text{OPT}(I)$. }

In this case, we first use that we can get a strictly better bound than $2z_v$ for all non-terminal vertices, where $z_v^P$ is bounded away from 0, i.e., for all vertices in $V_P$. 
For any vertex $v\in V_P$ we have $z_v^P \geq \frac{\delta}{16} >0$ and thus
\begin{align*}
z_v^P + 2e^{-z_v^P}&\le \begin{cases} 2 - f'(\delta), & \text{ if } z_v^P \leq 1,\\
z_v^P + \frac{2}{e}, & \text{ if } z_v^P \geq 1
\end{cases}\\ 
&\leq (2-f'(\delta))z_v
\end{align*}
for some strictly positive function $f'$. 

Hence, the total expected cost of $F_1$ is
\begin{align*}
 \mathbb{E}[c(F_1)] &\leq \sum_{v \not\in (S \cup T) } (z_v^P + 2\cdot e^{-z_v^P}) + \sum_{v \in S \cup T} z_v^P\\
 & \leq \sum_{v \not\in V_P} 2 z_v + \sum_{v \in V_P} (2-f'(\delta))z_v\\
  &\le 2\text{OPT}(I) -\frac{\delta}{16}\cdot \text{OPT}(I) \cdot f'(\delta)\\
&\le (2-f''(\delta))\cdot \text{OPT}(I). 
\end{align*}
Here we used the same estimation as in \eqref{eq:cost_of_tour}. Moreover, we exploited that $z_v \geq 1$ and that for all nodes in $V_P$, which are at least $\frac{\delta}{16}\text{OPT}(I)$ many, we will spend at least $f'(\delta)$ less than 2. \\

\noindent\textbf{Case 3.2: $|V_P| < \frac{\delta}{16}\text{OPT}(I)$} and \textbf{$\sum_{v \in S\cup T} z_{v} < \frac{\delta}{16}\text{OPT}(I)$}.

Intuitively, in this case we have large total weight on paths even though there are few vertices with large $z_v^P$. This means that there must be vertices with $z_v^P \ge 1$. To formalize this, we define 
$V_H = \{v\in V\setminus (S \cup T)\colon z_v^P\ge 1 \}$, and $H \coloneqq \sum_{v\in V_H} (z_v^P-1)$. 
We claim $H \ge \frac{\delta}{16}\text{OPT}(I)$.
We show this by contradiction. 
We can bound $D(I)$ as follows
\begin{align*}
D(I) &\leq \sum_{v\in V} z_{v}^P = \sum_{v\in (S\cup T)} z_{v}^P + \sum_{v\in V\setminus (S\cup T)} z_{v}^P \\ 
&< \frac{\delta}{16}\text{OPT}(I) + \sum_{v\in V\setminus (S\cup T\cup V_P)} z_{v}^P + \sum_{v\in V_P} z_{v}^P \\ 
& \le \frac{\delta}{16}\text{OPT}(I) + |V\setminus (S\cup T\cup V_P)| \cdot \frac{\delta}{16} + |V_P| + H\\
&< \frac{\delta}{16}\text{OPT}(I) +  \frac{\delta}{16}\text{OPT}(I) +\frac{\delta}{16}\text{OPT}(I) +\frac{\delta}{16}\text{OPT}(I)\\
& = \frac{\delta}{4}\text{OPT}(I), 
\end{align*}
which is a contradiction. 
Note that in the second inequality we used $\sum_{v \in S\cup T} z_{v}^P < \frac{\delta}{16}\text{OPT}(I)$ by assumption of this case and the fact that $z_v^P \le z_v$. Moreover, we exploited that $|V\setminus (S\cup T\cup V_P)| \leq n-k \leq \text{OPT}(I)$.

We now show that we can improve over $2\cdot \text{OPT}$ by the mass of $H$. First, observe that for $v \in V_H$
$$
z_v^P+2e^{-z_v^P}< z_v^P+1 \leq 2z_v -(z_v^P-1).$$
This means, again using the same estimations as in \eqref{eq:cost_of_tour} and $V_H \subseteq V \setminus (S \cup T)$, we have
\begin{align*}
    \mathbb{E}[c(F_1)] &= \sum_{v \not\in S \cup T} z_v^P + 2\cdot e^{-z_v^P} + \sum_{v \in S \cup T} z_v^P\\%
    &\leq \sum_{v \in V \setminus V_H} 2z_v + %
        \sum_{v \in V_H} \left( 2z_v -(z_v^P-1) \right)\\
    &\leq 2 \sum_{v \in V} \sum_{i \in [k]} z_{i,v} - H \leq \left(2-\frac{\delta}{16}\right)\text{OPT}(I).
\end{align*}


\end{proof}

\section{Conclusion}
We have shown that Algorithm~\ref{alg:main} is a $2$-approximation algorithm for Graphic Multi-Path TSP and that \LP has an integrality gap of at most $2$.
Our approach has not used the state-of-the-art algorithms for Graphic TSP.
While for a constant number of terminal pairs, the reduction of Traub and Vygen~\cite{Traub_Vygen_2024} leads to a $(1.4+\epsilon)$-approximation using the algorithm of Seb\H{o} and Vygen~\cite{sebHo2014shorter}, it is not obvious how these ideas can be applied if there is an arbitrary number of terminal pairs. We therefore leave as an open question to apply known techniques for Graphic TSP to the general Graphic Multi-Path TSP.

Notably, the worst case in our analysis appears when the sum of distances between the terminal pairs goes to zero. In \Cref{sec:reduction}, we show that an improvement over 2 for this special case, i.e., $s_i=t_i$ for all commodities $i \in [k]$, leads directly to an improvement over 2 for Graphic Multi-Path TSP. 

The reader may wonder if it is helpful to use branching decompositions~\cite{bang1995preserving} for Graphic Multi-Path TSP in the spirit of \cite{bohm2025approximating}. However, a straightforward use of such a decomposition leads to an approximation ratio worse than $2$ because all edges of a branching that are not located on the fundamental path between the terminals have to be doubled for parity correction and there is an additional cost for collecting uncovered vertices.

Finally, in \Cref{sec:ordered_TSP}, we use a very similar approach as in \Cref{alg:main} to provide a better than 1.791 approximation for Graphic Ordered TSP. 